%% file: neutralmain.tex
\documentclass[a4paper,USenglish,dvipsnames]{article}

\usepackage[utf8]{inputenc}
\usepackage{fullpage}
\usepackage{amssymb}
\usepackage{amsmath}
\usepackage{amsthm}

\input{macros.tex}
\usepackage[disable]{todonotes}
\usepackage{xspace}

\usepackage{tikz}
\usetikzlibrary{math, automata, positioning, patterns, decorations.pathreplacing}

\newtheorem{remark}{Remark}
\newtheorem{lemma}{Lemma}
\newtheorem{theorem}{Theorem}
\newtheorem{proposition}{Proposition}
\newtheorem{example}{Example}
\newtheorem{definition}{Definition}
\newtheorem{corollary}{Corollary}

\title{Time for Timed Monitorability}
\author{
Thomas M.\ Grosen (Aalborg University, Aalborg, Denmark)\\
Sean Kauffman (Queen's University, Kingston, Canada)\\
Kim G.\ Larsen (Aalborg University, Aalborg, Denmark)\\
Martin Zimmermann (Aalborg University, Aalborg, Denmark)}
\date{}

\begin{document}

\maketitle

\begin{abstract}
Monitoring is an important part of the verification toolbox, in particular in situations where exhaustive verification using, e.g., model-checking is infeasible.
The goal of online monitoring is to determine the satisfaction or violation of a specification  during runtime, i.e., based on finite execution prefixes. 
However, not every specification is amenable to monitoring, e.g., properties for which no finite execution can witness satisfaction or violation. 
Monitorability is the question of whether a given specification is amenable to monitoring, and has been extensively studied in discrete time.

Here, we study the monitorability problem for real-time properties expressed as Timed Automata. For specifications given by deterministic Timed Muller Automata, we prove decidability while we show that the problem is undecidable for specifications given by nondeterministic Timed Büchi automata. 

Furthermore, we refine monitorability to also determine bounds on the number of events as well as the time that must pass before monitoring the property may yield an informative verdict.
We prove that for deterministic Timed Muller automata, such bounds can be effectively computed.  In contrast we show that for nondeterministic Timed Büchi automata such bounds are not computable.
\end{abstract}

%%%%%%%%%%%%%%%%%%%%%%%%%%%%%%%%%%%%%%%%%%%%%%%%
%%%%%%%%%%%%%%%%%%%%%%%%%%%%%%%%%%%%%%%%%%%%%%%%
%%%%%%%%%%%%%%%%%%%%%%%%%%%%%%%%%%%%%%%%%%%%%%%%
\section{Introduction}

A fundamental challenge in Runtime Verification (RV) is that many properties provide no utility when monitored in an online setting.
Thus, it is of utmost importance to identify those properties that do provide useful information. 

Behaviors of long-running systems are typically specified as languages of infinite words, but online monitors only observe finite (prefixes of) system executions.
Thus, a monitor has to determine whether such a finite prefix already implies satisfaction or violation of the property. 
While many different types of monitors have been proposed, most online monitors return information in the form of \emph{verdicts} about the finite prefix.
Inherently to the problem, there are at least three verdicts~\cite{bauer2011runtime}: $\{\top,\bot,\unknown\}$, where $\top$ and $\bot$ are \emph{conclusive} verdicts that mean that the finite prefix guarantees that \emph{every} possible infinite extension satisfies, respectively violates, the property, and the \emph{inconclusive} verdict~$\unknown$ signifying that neither is the case. 

For example, consider an arbiter granting access to a shared resource. 
The property~$\varphi_1$ expressing \myquot{there is no request during the first second} is satisfied if the first request arrives after two seconds, no matter how the execution continues.
Hence, the verdict for such a finite execution is $\top$.
Similarly, the property~$\varphi_2$ expressing \myquot{after any request, there are no requests for at least one second} is violated as soon as two consecutive requests are observed within one second, no matter how the execution continues.
Hence, the verdict for such a finite execution is $\bot$. 
On the other hand, if requests arrive with a gap of two seconds between them, then the verdict for such a prefix (w.r.t.\ $\varphi_2$) is $\unknown$, since there are infinite extensions satisfying the property and infinite extensions violating it.

As seen above, there are prefixes for which we get a conclusive verdict w.r.t.\ $\varphi_2$.
This distinguishes it from properties like $\varphi_3$ expressing \myquot{every request is eventually granted}, for which every finite prefix can be extended to satisfy the property and can be extended to violate it.
Hence, every finite prefix yields the verdict~$\unknown$.
Phrased concisely: monitoring the property~$\varphi_3$ is futile.

However, for the user it is not transparent, while receiving the verdict~$\unknown$, whether in the future a conclusive verdict may be given, or whether every possible extension yields the verdict~$\unknown$.
The concept of \emph{monitorability} has been introduced to capture those properties that are amenable to monitoring.
It comes in two variants, strong monitorability (every prefix can be extended to one that yields a conclusive verdict) and weak monitorability (some prefix yields a conclusive verdict).
In this language, $\varphi_3$ is not weakly monitorable while $\varphi_1$ and $\varphi_2$ are (even strongly) monitorable. 
Thus, before constructing and deploying a monitor for a property, it is prudent to first check whether the property is monitorable. 

The problem of deciding if a property is monitorable has been studied extensively over the last 20 years (we discuss related work in Section~\ref{sec:related}), but only in the setting of discrete time until very recently. 
However, many properties require real-time constraints to express, e.g., deadlines like \myquot{every request is answered within 745 ms}.
In particular, safety-critical systems are nearly always real-time systems with physics-based deadlines, and these systems tend to benefit the most from formal verification methods like RV.
Had the property~\myquot{the Therac~25 control program must wait eight seconds before switching between X-ray and electron modes} been monitored, it might have saved lives~\cite{leveson1993therac}.
While monitoring algorithms exist for real-time properties~\cite{thati2005monitoring,bauer2006monitoring,basin2012algorithms,DBLP:conf/formats/GrosenKLZ22}, the problem of real-time monitorability has gone largely unexamined.

\subparagraph{Our Contribution}
This work makes monitoring of real-time systems more useful by examining the monitorability problem for real-time properties and by introducing qualitative refinements of the verdict~$\unknown$.
We consider real-time properties expressed by Timed Automata~(TA) over infinite words~\cite{alur1994tba}.
Nondeterministic Timed Büchi Automata (TBA) are used for model checking tools like \textsc{Uppaal}~\cite{larsen2004online} and for monitoring temporal logics~\cite{DBLP:conf/formats/GrosenKLZ22,delay,assumptions} and are strictly more expressive than deterministic Timed Muller Automata (DTMA).
We prove that strong and weak monitorability are undecidable for nondeterministic TBA, but decidable for DTMA.
Thus, one can algorithmically determine that it is futile to monitor properties like $\varphi_3$, thereby increasing the applicability of monitoring of real-time systems.

Furthermore, we introduce monitorability with step-bounded horizons that strengthens monitorability by limiting the number of events in a timed-word before a conclusive verdict must be reached.
A step-bounded horizon allows one to determine for a given property and $n \in \nats$, if a conclusive verdict is possible within $n$ steps, enabling corrective actions earlier.
Again, we show that monitorability with step-bounded horizons is undecidable for nondeterministic TBA, but decidable for DTMA.

Finally, we refine monitoring of real-time properties with time-horizon verdicts.
Here, the  verdict~$\unknown$ is enhanced with information about the minimum time until a conclusive verdict may be reached.
Intuitively, before this time is reached, the monitor will only yield the inconclusive verdict~$\unknown$, i.e., no information can be gained from querying the monitor before.
This notion was introduced by Grosen et al.~\cite{DBLP:conf/formats/GrosenKLZ22} as ``time-predictive'' monitoring.
Here, we formally prove that time-horizon queries can be computed effectively for DTMA.

Thus, our results highlight the importance of properties being given by deterministic timed automata when checking monitorability. This is in contrast to monitoring itself, where it suffices to have nondeterministic automata for the property \emph{and} its negation~\cite{DBLP:conf/formats/GrosenKLZ22}, which is, e.g., the case when specifying properties in Metric Interval Temporal Logic~\cite{alur1996mitl}. 
Finally, we show that it is necessary to have automata for the property and the complement, as the monitoring function is otherwise not effectively computable.

%%%%%%%%%%%%%%%%%%%%%%%%%%%%%%%%%%%%%%%%%%%%%%%%
%%%%%%%%%%%%%%%%%%%%%%%%%%%%%%%%%%%%%%%%%%%%%%%%
%%%%%%%%%%%%%%%%%%%%%%%%%%%%%%%%%%%%%%%%%%%%%%%%
\section{Preliminaries}
\label{sec:prelims}

The nonnegative integers are denoted by $\nats$ and the nonnegative reals by $\nnreals$. An alphabet is a finite nonempty set of letters.

A timed word is a pair~$\rho= (\sigma, \tau)$, where $\sigma$ is a (finite or infinite) word over an alphabet~$\Sigma$ and $\tau$ is a sequence of non-decreasing, non-negative real numbers of the same length as $\sigma$. 
For convenience, we often write $(\sigma_1, \tau_1) (\sigma_2, \tau_2) \cdots $ for a timed word~$(\sigma, \tau)$.
$\finwords$ and $\infwords$ denote the sets of finite and infinite timed words over $\Sigma$. For $n\in\nats\cup\{\infty\}$ we denote by $T\Sigma^{\leq n}$ the set of timed words over $\Sigma$ of  length at most $n$.
Given a finite word $\rho = (\sigma_1, \tau_1) (\sigma_2, \tau_2) \cdots   (\sigma_n, \tau_n)$, we denote its duration as $\tau(\rho) = \tau_n$. 
Slightly abusively, we write $\emptyword$ for the empty timed word~$(\emptyword, \emptyword)$ and define $\tau(\emptyword) = 0$.
For a finite timed word~$\rho = (\sigma_1, \tau_1) (\sigma_2, \tau_2) \cdots  (\sigma_n, \tau_n)$, a finite or infinite word~$\rho' = (\sigma_1', \tau_1') (\sigma_2', \tau_2')\cdots $ and a timepoint~$t \ge \tau(\rho)$, we define the concatenation of $\rho$ and $\rho'$ at $t$ as 
\[\rho \conc{t} \rho' = (\sigma_1, \tau_1) (\sigma_2, \tau_2)\cdots  (\sigma_n, \tau_n) (\sigma_1', \tau_1' + t) (\sigma_2', \tau_2' + t) \cdots \]
which is a timed word.
As a shorthand, we write $\rho \cdot{} \rho'$ for $\rho \conc{\tau(\rho)} \rho'$.
Given two finite words~$\rho,\rho'$ and a timepoint~$t \ge \tau(\rho)$, we write $\rho \prefixrel_t \rho'$ if there exists a $\rho''$ such that $\rho \cdot_t \rho'' = \rho'$.

\begin{example}
    Consider the words $\rho = (a, 0) (b, 10)$ and $\rho' = (c, 5) (d, 15)$. We have
    \begin{itemize}
        \item $\rho \conc{} \rho' = \rho \conc{10} \rho' = (a, 0) (b, 10) (c, 15) (d, 25)$,
        \item $\rho \conc{15} \rho' = (a, 0) (b, 10) (c, 20) (d, 30)$, and
        \item $\rho' \conc{} \rho = \rho' \conc{15} \rho = (c, 5) (d, 15) (a, 15) (b, 25)$.
    \end{itemize}
\end{example}

%%%%%%%%%%%%%%%%%%%%%%%%%%%%%%%%%%%%%%%%%%%%%%%%
%%%%%%%%%%%%%%%%%%%%%%%%%%%%%%%%%%%%%%%%%%%%%%%%
%%%%%%%%%%%%%%%%%%%%%%%%%%%%%%%%%%%%%%%%%%%%%%%%
\subparagraph*{Timed Automata}
A timed Büchi automaton (TBA) is a tuple $\aut = (Q, Q_0, \Sigma, \clocks, \Delta, F)$ where $Q$ is a finite set of locations, $Q_0 \subseteq Q$ is the set of initial locations, $\Sigma$ is an alphabet, \clocks is a finite set of clocks, $F \subseteq Q$ is a set of accepting locations, and $\Delta \subseteq Q \times Q \times \Sigma \times 2^\clocks \times G(\clocks)$ is a set of transitions, where $G(\clocks)$ is the set of clock constraints over \clocks.
A transition $(q, q', \alpha, \lambda, g) \in \Delta$ is an edge from $q$ to $q'$ with label $\alpha \in \Sigma$, where $\lambda \in 2^\clocks$ is a set of clocks to be reset and $g \in G(\clocks)$ is a clock constraint.
A clock constraint is a finite conjunction of atomic constraints of the form $x \sim n$ where $x \in \clocks$, $n \in \nats$, and $\sim\, \in \{<, \le, =, \ge, >\}$.
A state of $\aut$ is a pair $(q, v)$ where $q \in Q$ and $v \colon \clocks \rightarrow \nnreals$ is a clock valuation.

A run of $\aut$ from a state $(q_0, v_0)$ over an infinite word
% \todo[color=green]{we maybe need to define runs over finite-words} 
$(\sigma, \tau) \in \infwords$ is a sequence of steps of the form
\[(q_0, v_0) \xrightarrow{(\sigma_1, \tau_1)} (q_1, v_1) \xrightarrow{(\sigma_2,\tau_2)} (q_2,v_2) \xrightarrow{(\sigma_3,\tau_3)} \cdots\]
where for every $i \ge 1$, there is a transition~$(q_{i-1}, q_i, \sigma_i, \lambda_i, g_i) \in \Delta$ such that $v_i(x) = 0$ for all $x \in \lambda_i$ and $v_i(x) = v_{i-1}(x) + \tau_i - \tau_{i-1}$ otherwise, and $g$ is satisfied by $v_{i-1} + \tau_i - \tau_{i-1}$ (where we use $\tau_0 = 0$).
For a run \textit{r}, $\textit{inf}(r) \subseteq Q$ denotes the set of locations visited infinitely often in $r$. 
A run $r$ is (Büchi) accepting if $\textit{inf}(r) \cap F \ne \emptyset$.

A timed Muller automaton (TMA) $\aut = (Q, Q_0, \Sigma, \clocks, \Delta, \mathcal{F})$ is like a TBA, but the set~$F$ of accepting locations is replaced by a set~$\mathcal{F} \subseteq 2^Q$ of sets of locations.
A run $r$ of $\aut$ is (Muller) accepting if $\textit{inf}(r) \in \mathcal{F}$.

The language~$L(\aut)$ of a timed (Büchi or Muller) automaton is the set of words~$\rho \in \infwords$ such that $\aut$ has an accepting run over $\rho$.

An automaton is deterministic if the set of initial locations is a singleton and if all edges from the same location and with the same label must have disjoint clock constraints. We use the abbreviations DTBA and DTMA to refer to the two deterministic automaton models.

\begin{proposition}
The following results on TBA and TMA are due to Alur and Dill~\cite{alur1994tba}.
\begin{enumerate}
    \item Let $\aut = (Q, Q_0, \Sigma, \clocks, \Delta, F)$ be a TBA and let $\aut' = (Q, Q_0, \Sigma, \clocks, \Delta, \{F' \mid F \cap F' \neq \emptyset\})$, which is a TMA. Then, $L(\aut) = L(\aut')$. Furthermore, if $\aut$ is deterministic, then so is $\aut'$.
    \item For every TMA~$\aut = (Q, Q_0, \Sigma, \clocks, \Delta, \mathcal{F})$, there is a TBA~$\aut'$ with $L(\aut) = L(\aut')$. The set of locations of $\aut'$ has the form~$Q \times \mathcal{F} \times \{0, 1, \ldots, |Q|\}$.
    \item The class of languages accepted by DTMA is a strict subset of the class of languages accepted by TBA.
\end{enumerate}
\end{proposition}

%%%%%%%%%%%%%%%%%%%%%%%%%%%%%%%%%%%%%%%%%%%%%%%%
%%%%%%%%%%%%%%%%%%%%%%%%%%%%%%%%%%%%%%%%%%%%%%%%
%%%%%%%%%%%%%%%%%%%%%%%%%%%%%%%%%%%%%%%%%%%%%%%%
\subparagraph*{Logic}
We use Metric Interval Temporal Logic (MITL) to formally express properties to be monitored.
The syntax of MITL formulas over a finite alphabet~$\Sigma$ is defined as
\[
\varphi ::= p \mid  \neg \varphi \mid  \varphi \vee \varphi \mid  X_I \varphi \mid  \varphi\ U_{I} \varphi\,
\]
where $p \in \Sigma$ and $I$ ranges over non-singular intervals over $\nnreals$ with endpoints in ${\nats \cup \{\infty\}}$.
Note that we often write $\sim\,n$ for $I=\{d\in\nnreals\mid d\sim n\}$ where $\sim\ \in \{<,\leq,\geq,>\}$, and $n \in \nats$.
  We also define the standard syntactic sugar
$\true = p \vee \neg p$,
$\false = \neg \true$,
$\varphi \wedge \psi = \neg (\neg \varphi \vee \neg \psi)$,
$\varphi \rightarrow \psi = \neg \varphi \vee \psi$,
$F_{I} \varphi = \true\ U_{I} \varphi$, 
and $G_{I} \varphi = \neg F_{I} \neg \varphi$.

The semantics of MITL is defined over infinite timed words.
Given such a timed word~$\rho = (\sigma_1,\tau_1)(\sigma_2,\tau_2)\cdots  \in \infwords$, a position $i \geq 1$, and an MITL formula~$\varphi$, we inductively define the satisfaction relation $\rho, i \models \varphi$ as  follows:
\begin{itemize}
    \item $\rho,i \models p $  if and only if $p = \sigma_i$.
    \item $\rho,i \models \neg \varphi$ if and only if $\rho,i \not\models \varphi$.
    \item $\rho,i \models \varphi \vee \psi$  if $\rho,i \models \varphi$ or $\rho,i \models \psi$.
    \item $\rho,i \models X_I \varphi$ if and only if $\rho,(i+1) \models \varphi$ and $\tau_{i+1} - \tau_i \in I$.
    \item $\rho,i \models \varphi\ U_{I} \psi$ if and only if  there exists $k \geq i$ s.t.\ $\rho,k \models \psi$, $\tau_k - \tau_{i} \in I$, and $\rho,j \models \varphi$ for all $i \leq j < k$. 
\end{itemize}
  We write $\rho \models \varphi$ whenever $\rho, 1 \models \varphi$, and say that $\rho$ satisfies $\varphi$.
The language $\lang(\varphi)$ of an MITL formula~$\varphi$ is the set of all infinite timed words that satisfy~$\varphi$.

\begin{proposition}[\cite{alur1996mitl,brihaye2017mightyl}]
\label{prop:mitl2tba}
For each MITL formula $\varphi$ there is a TBA $\aut$ such that $\lang(\varphi) = \lang(\aut)$.
\end{proposition}  

\begin{example}\label{ex:mitl2automaton}
Figure~\ref{fig:aut-example} illustrates the above proposition by providing a DTBA over the alphabet~$\Sigma = \{a,b,c\}$ for the formula~$F_{[0,10]}a \land G_{[0,20]}\neg b$ and its negation.
\end{example}

\begin{figure}[ht]
    \centering
    \begin{tikzpicture} [node distance = 2cm,thick]
    \node (q0)     [state, initial text={}, minimum size =15]          {};
    \node (i)   at (.8,0.5) {};
    % \node (i)   at (-0.5,-1) {};
    \node (q1)     [state, right = of q0, minimum size =15]    {};
    \node (phi)    [state, right = of q1, minimum size =15]    {$\varphi$};
    \node (notphi) [state, left = of q0, minimum size =15, inner sep =1]    {$\neg\varphi$};
    
    \draw[->, > = stealth, thick] (i) edge (q0);
     \draw[->, > = stealth, thick] (q0) edge node[above] {$a$} node[below]{$x \leq 10$} (q1);
     \draw[->, > = stealth, thick] (q1) edge node[above] {$a, b,c$} node[below]{$x > 20$} (phi);

     \draw[->, > = stealth, thick] (q0) edge[bend right=50] node[above]{$b$} (notphi);
     \draw[->, > = stealth, thick] (q0) edge[] node[above]{$a,c$} node[below]{$x > 10$} (notphi);

     % \draw[->, > = stealth, thick] (q1) edge[bend left=55] node[above]{$b$}  node[below]{$x\leq20$} (notphi);
     \draw[->, > = stealth, thick] (q1) edge[bend left=40] node[above]{$b$}  node[below]{$x\leq20$} (notphi);
    
    \draw[->, > = stealth, thick] (phi) edge [loop below] node[] {$a,b,c$} (phi);
    \draw[->, > = stealth, thick] (q1) edge [loop above] node[above,align=center] {$a,c$\\$x \leq 20$} (q1);
    \draw[->, > = stealth, thick] (notphi) edge[loop below] node[]{$a,b,c$} (notphi);

    \draw[->, > = stealth, thick] (q0) edge[loop above] node[above,align=center]{$c$\\$x \leq 10$} (q0);

\end{tikzpicture}
    \caption{DTBA for the language of the MITL formula~$\varphi = F_{[0,10]}a \land G_{[0,20]}\neg b$ and its negation: If location~$\varphi$ ($\neg \varphi$) is accepting then it accepts $L(\varphi)$ ($L(\neg \varphi)$, respectively).} 
    \label{fig:aut-example}
\end{figure}

%%%%%%%%%%%%%%%%%%%%%%%%%%%%%%%%%%%%%%%%%%%%%%%%
%%%%%%%%%%%%%%%%%%%%%%%%%%%%%%%%%%%%%%%%%%%%%%%%
%%%%%%%%%%%%%%%%%%%%%%%%%%%%%%%%%%%%%%%%%%%%%%%%
\subparagraph*{Monitoring}
The monitoring problem asks to make verdicts about the satisfaction or violation of properties (over infinite timed words) after having observed only a finite prefix. 
Here, we follow the classical approach of considering the possible extensions of a finite observations.

\begin{definition}[Observation]
    An observation is a pair~$(\rho, t)$ containing a finite timed word~$\rho$ and a timepoint~$t \ge \tau(\rho)$, representing the current timepoint (which might be later than the last observed event in $\rho$). As we use observations as inputs for algorithms, we require that all timepoints in $\rho$ and $t$ are rational.
\end{definition}

We continue with giving some intuition for the three-valued monitoring approach.
Here, one is given an observation and aims to determine whether a property~$\varphi$ (of infinite timed words) is already satisfied, already violated, or neither.
\begin{itemize}
    \item If all possible extensions of the observation satisfy $\varphi$, then we give the corresponding verdict~$\top$ signifying that the observation conclusively witnesses satisfaction of $\varphi$.
    \item If all possible extensions of the observation violate $\varphi$, then we give the corresponding verdict~$\bot$ signifying that the observation conclusively witnesses violation of $\varphi$.
    \item Otherwise, i.e., if there is an extension of the observation that satisfies $\varphi$ and there is a extension of the observation that violates $\varphi$, then we give the inconclusive verdict~$\unknown$.
\end{itemize}

Let us formalize this intuition.

\begin{definition}[Timed Monitoring Function]
Given a property~$\varphi \subseteq \infwords$ and an observation~$(\rho, t)$, the monitoring function~\monitoring{\varphi} is defined as
    \[
    \monitoring{\varphi}(\rho, t) = \begin{cases}
        \top & \textit{ if }  \rho \conc{t} \mu \in \varphi \text{ for all } \mu \in \infwords,\\
        \bot & \textit{ if }  \rho \conc{t} \mu \not\in \varphi \text{ for all } \mu \in \infwords,\\
        \unknown & \textit{ otherwise.}
    \end{cases}
    \]
In the following, we use $\monitoring{\varphi}(\rho)$ as a shorthand for $\monitoring{\varphi}(\rho, \tau(\rho))$.
\end{definition}

\begin{example}
\label{example_monitoring}
Consider the specification~$\varphi  = F_{[0,10]}a \land G_{[0,20]}\neg b$ from Example~\ref{ex:mitl2automaton}. We have
\begin{itemize}
    \item $\monitoring{\varphi}((a,3), 4) = \unknown$,
    \item $\monitoring{\varphi}((a,11), 11) = \bot$,
    \item $\monitoring{\varphi}((a,3)(c, 7), 13) = \unknown$,  
    \item $\monitoring{\varphi}((a,3)(c, 7)(c,22), 22) = \top$, but
    \item $\monitoring{\varphi}((a,3)(c, 7)(b,12), 12) = \bot$. Also,
    \item $\monitoring{\varphi}((a,3)(c, 7), 22) = \top$ while 
    \item $\monitoring{\varphi}((a,3)(c, 7)) = \monitoring{\varphi}((a,3)(c, 7), 7) = \unknown$, i.e., the current timepoint~$t$ can yield conclusive verdicts when time is passing, even if no new events are observed.
\end{itemize}
\end{example}

\begin{proposition}[Effectiveness of Timed Monitoring \cite{DBLP:conf/formats/GrosenKLZ22}]
\monitoring{\varphi} is effectively computable by a zone-based online algorithm\footnote{See Page~\pageref{def:zones} for a formal definition of zones.} that requires TBA for \emph{both} $\varphi$ and $\infwords \setminus \varphi$.
\end{proposition}

Recall that TBA are \emph{not} closed under complement~\cite{alur1994tba}, so this result is not applicable to all properties accepted by TBA.
However, for the important special case of properties~$\varphi$ specified in MITL, \monitoring{\varphi} is effectively computable, as MITL properties are closed under complementation (as the logic allows for negations) and MITL formulas can be translated into equivalent TBA (see Proposition~\ref{prop:mitl2tba}).
Similarly, when $\varphi$ is given by a DTMA, then \monitoring{\varphi} is effectively computable, as DTMA are closed under complementation and can be turned into equivalent TBA~\cite{alur1994tba}.

However, it was previously open whether \monitoring{\varphi} was computable if $\varphi$ was given by a non-deterministic automaton, but we did not have access to an automaton for the complement~$\infwords \setminus \varphi$.
In Section~\ref{sec:tba}, we answer the question negatively.

%%%%%%%%%%%%%%%%%%%%%%%%%%%%%%%%%%%%%%%%%%%%%%%%
%%%%%%%%%%%%%%%%%%%%%%%%%%%%%%%%%%%%%%%%%%%%%%%%
%%%%%%%%%%%%%%%%%%%%%%%%%%%%%%%%%%%%%%%%%%%%%%%%
\subparagraph*{Monitorability}

Not every property is amenable to monitoring, e.g., for $\varphi = G_{\ge 0}F_{\ge 0}a$, we have $\monitoring{\varphi}(\rho,t)=\unknown$ for every observation~$(\rho,t)$.
The reason is that every $\rho$ can be extended to satisfy $\varphi$ and can be extended to violate $\varphi$.
In the untimed setting, much effort has been put into characterizing the monitorable properties, i.e., those for which monitoring can generate some information.
Here, we consider monitorability in the timed setting.

\begin{definition}[Timed Monitorability]
\label{def:monitorability}
Fix an observation~$(\rho,t)$ and a property~$\varphi \subseteq \infwords$. 
\begin{itemize}

    \item $\varphi$ is strongly $(\rho,t)$-monitorable if and only if 
    \[
    \text{for all } \rho' \in \finwords  \text{ there exists } \rho'' \in \finwords \text{ such that } \monitoring{\varphi}(\rho \conc{t} \rho' \cdot \rho'') \in \{\top, \bot\}.
    \]

    \item $\varphi$ is weakly $(\rho,t)$-monitorable if and only if
\[
    \text{there exists }\rho'' \in \finwords \text{ such that } \monitoring{\varphi}(\rho \conc{t} \rho'') \in \{\top,\bot\}.
  \]
      \item $\varphi$ is strongly monitorable if it is strongly $(\varepsilon,0)$-monitorable.

    \item $\varphi$ is weakly monitorable if it is weakly $(\varepsilon,0)$-monitorable.

\end{itemize}
\end{definition}

\begin{example}\hfill

\begin{enumerate}
    \item Consider the property~$\varphi_1 = F_{\ge 0} a$.
    For every observation~$(\rho, t)$, we have $\monitoring{\varphi_1}(\rho \conc{t} (a,0),t)  = \top$.
    Hence, $\varphi_1$ is strongly monitorable.

    \item Now, consider $\varphi_2 = a \rightarrow G_{\ge 0}F_{\ge 0} a$. 
    Then, we have $\monitoring{\varphi_2}((b, 0),0) = \top$, as every extension of $((b, 0),0)$ satisfies $\varphi_2$ (as the premise is violated).
    Hence, $\varphi_2$ is weakly monitorable.

    However, it is not strongly monitorable: Consider the observation~$((a, 0),0)$, for which every extension satisfies the premise. We have $(a, 0)\conc{0} \rho''\conc{} (a,0)(a,1)(a,2)\cdots  \models \varphi_2$ and $(a, 0)\conc{0} \rho''\conc{} (b,0)(b,1)(b,2)\cdots  \not\models \varphi_2$.
    Hence, $\monitoring{\varphi_2}((a, 0)\conc{0} \rho'')=\unknown$ for all $\rho''$.

    \item Now, consider $\varphi_3 = G_{\ge 0}F_{\ge 0} a$. Arguments as for $\varphi_2$ show that it is neither strongly nor weakly monitorable, as every finite word can be extended by $(a,0)(a,1)(a,2)\cdots $ to satisfy $\varphi_3$ and can be extended by $(b,0)(b,1)(b,2)\cdots $ to violate $\varphi_3$.
        
\end{enumerate}
\end{example}

\begin{remark}
The astute reader might wonder why we quantify only over words~$\rho'$ (and $\rho''$) in Definition~\ref{def:monitorability} and not over words and timepoints to concatenate at. 
The reason is that both definitions are equivalent, as $\rho_1 \cdot_t \rho_2$, for finite words $\rho_1$ and $\rho_2$ and a timepoint~$t \ge \tau(\rho)$, is equal to $\rho_1 \cdot \rho_2'$, where $\rho_2'$ is obtained from $\rho_2$ by incrementing all its timepoints by $t - \tau(\rho_1)$.
\end{remark}

Let us continue by collecting some simple consequences of Definition~\ref{def:monitorability}.

\begin{remark}
\label{rem_properties}
Let $(\rho,t)$ and $(\rho',t')$ be two observations with $\rho \prefixrel_t \rho'$, and let $\varphi \subseteq \infwords$.
\begin{enumerate}
    \item If $\varphi$ is strongly $(\rho, t)$-monitorable, then it is also weakly $(\rho, t)$-monitorable. Thus, if $\varphi$ is strongly monitorable, then it is also weakly monitorable.
    \item If $\varphi$ is strongly $(\rho,t)$-monitorable, then it is also strongly $(\rho',t')$-monitorable.
    \item If $\varphi$ is not weakly $(\rho,t)$-monitorable, then it is also not weakly $(\rho',t')$-monitorable.     
\end{enumerate}
\end{remark}

In the following we study the decidability of monitorability, i.e., we consider the following decision problems where properties are given by timed automata:
\begin{enumerate}
    \item\label{decprop_strong} Given a property~$\varphi$, is $\varphi$ strongly monitorable?
    \item\label{decprop_weak} Given a property~$\varphi$, is $\varphi$ weakly monitorable?
    \item\label{decprop_strongvar} Given a property~$\varphi$ and an observation~$(\rho,t)$, is $\varphi$ strongly $(\rho,t)$-monitorable?
    \item\label{decprop_weakvar} Given a property~$\varphi$ and an observation~$(\rho,t)$, is $\varphi$ weakly $(\rho,t)$-monitorable?
\end{enumerate}

By definition, if Problem~\ref{decprop_strongvar} is decidable for a class of properties, then Problem~\ref{decprop_strong} is also decidable for the same class of properties.
Hence, if Problem~\ref{decprop_strong} is undecidable for a class of properties, then Problem~\ref{decprop_strongvar} is also undecidable for the same class of properties.
A similar relation holds between Problem~\ref{decprop_weak} and Problem~\ref{decprop_weakvar}.

%%%%%%%%%%%%%%%%%%%%%%%%%%%%%%%%%%%%%%%%%%%%%%%%
%%%%%%%%%%%%%%%%%%%%%%%%%%%%%%%%%%%%%%%%%%%%%%%%
%%%%%%%%%%%%%%%%%%%%%%%%%%%%%%%%%%%%%%%%%%%%%%%%
\section{Monitoring and Monitorability for TBA}
\label{sec:tba}

In this section, we prove that the monitoring function is not computable and that the strong types of monitorability are undecidable when the property is given by a TBA.
Note that this shows that the positive results for monitoring in the literature~\cite{DBLP:conf/formats/GrosenKLZ22,delay,assumptions}, which require automata both for the property and its complement, are tight in that sense: Only giving an automaton for the property, but not for its complement, is not sufficient to compute the monitoring function.
We start by investigating the computability of the monitoring function.

\begin{theorem}[Ineffectiveness of Timed Monitoring]
\label{thm_monitpringundecidableforNTBA}
    The function~\myquot{Given a TBA~$\aut$ and an observation~$(\rho, t)$, return $\monitoring{L(\aut)}(\rho, t)$} is not computable.
\end{theorem}

\begin{proof}
For every property~$\varphi \subseteq \infwords$, we have $\monitoring{\varphi}(\emptyword,0) = \top$ if and only if $\varphi = \infwords$.
Hence, universality of a TBA~$\aut$ reduces to checking whether $\monitoring{L(\aut)}(\emptyword,0) = \top$.
As universality for timed automata is undecidable~\cite{alur2004decision}, the monitoring function cannot be computable.
\end{proof}

Note that the specification automaton~$\aut$ is part of the input in the problem considered in Theorem~\ref{thm_monitpringundecidableforNTBA}.
We leave it open whether $\monitoring{L(\aut)}$ is computable for every fixed $\aut$, i.e., in the setting where only the observation is the input.

Next, we turn our attention to deciding monitorability for TBA.

\begin{theorem}
\label{thm:undecmoni}
    Strong and weak monitorability are undecidable for properties given by TBA.
\end{theorem}

\begin{proof}
Strong monitorability for untimed non-deterministic Büchi Automata is \pspace-complete~\cite{diekert2012distributed}, which can be shown by a reduction from the universality problem~\cite{diekert2015note}. 
Analogously, we reduce the (undecidable~\cite{alur1994tba}) universality problem for non-deterministic timed automata (over finite words) to the problem of strong monitorability, following Diekert, Muscholl, and Walukiewicz~\cite{diekert2015note}.

Let $\aut = (Q,Q_0,\Sigma,\clocks, \Delta, F)$ be such a timed automaton, i.e., a finite run is accepting if it ends in a location in $F$.
We assume w.l.o.g.\ that $\aut$ is complete in the sense that every word has a run. 
This can always be achieved by adding a fresh sink location and by rerouting all missing transitions to it.

We add a new letter~$\Bsymbol \notin \Sigma$ to obtain $\Sigma_\Bsymbol = \Sigma \cup \{\Bsymbol\}$ and construct a TBA~$\aut' = (Q',Q_0,\Sigma_\Bsymbol,\clocks, \Delta', F')$ from $\aut$ such that $\aut'$ is strongly monitorable if and only if $\lang(\aut) = \finwords$.
To this end, we introduce three new locations~$\Dstate,\Estate,\Fstate$, i.e., $Q' = Q \cup \{\Dstate,\Estate,\Fstate\}$.
Next, we define $\Delta'$ (see Figure~\ref{fig:undec-automatastrong}) by copying the transitions from $\Delta$ and by adding the following transitions, where we write $q \xrightarrow{a} q'$ to denote $(q,q',a,\emptyset,\true)$: 
\begin{itemize}
    \item $\{ q \xrightarrow{\Bsymbol} \Dstate \mid q \in Q \setminus F\}$: from every non-accepting location of $\aut$, there is a $\Bsymbol$-transition to $\Dstate$.
    \item $\{ \Dstate \xrightarrow{a} \Estate, \Estate \xrightarrow{a} \Estate \mid a \in \Sigma \}$: for every $a \neq \Bsymbol$ there is an $a$-transition from $\Dstate$ to $\Estate$ and an $a$-labeled self-loop on $\Estate$.
    \item $\{ \Estate \xrightarrow{\Bsymbol} \Dstate, \Dstate \xrightarrow{\Bsymbol} \Dstate \}$: there is a $\Bsymbol$-transition from $\Estate$ to $\Dstate$ and a $\Bsymbol$-labeled self-loop on $\Dstate$.
    \item $\{ q \xrightarrow{\Bsymbol} \Fstate \mid q \in F\}$: from every accepting location there is a $\Bsymbol$-transition to $\Fstate$.
    \item $\{ \Fstate \xrightarrow{a} \Fstate \mid a \in \Sigma_\Bsymbol \}$: for every letter in $\Sigma_\Bsymbol$, there is a self-loop on $\Fstate$ labeled with that letter.
\end{itemize} 
We define the accepting locations of $\aut'$ as $F' = Q \cup \{\Dstate,\Fstate\}$.

\begin{figure}[ht]
    \centering

\begin{tikzpicture}[thick]
  % Draw the main wide rectangle with rounded corners
  \draw[thick,rounded corners] (0,0.2) rectangle (6,1.8);
  
  % Draw the smaller rectangle inside on the right
  \draw[thick,rounded corners] (4,0.2) rectangle (6,1.8);

    \node[anchor=west] at (0,1.45) {$Q\setminus F$};
    \node[anchor=west] at (4,1.45) {$F$};

    % Add automaton states
  % State on the right
  \node[state, accepting] (t) at (7.5,1) {$\Fstate$};

  % States on the left
  \node[state, accepting] (r) at (-1.5,1) {$\Dstate$};
  \node[state] (s) at (-3.5,1) {$\Estate$};

  \path[->, > = stealth]
  (5,1) edge node[above, near end] {$\Bsymbol$} (t)
  (t) edge[loop above] node[above] {$\Sigma_\Bsymbol$} ()
  (1,1) edge node[above, near end] {$\Bsymbol$} (r)
  (r) edge[bend left] node[below] {$\Sigma$} (s)
  (s) edge[bend left] node[above] {$\Bsymbol$} (r)
  (r) edge[loop above] node[above] {$\Bsymbol$} ()
  (s) edge[loop above] node[above] {$\Sigma$} ()
  ;

\end{tikzpicture}
        \caption{The TBA $\aut'$ constructed for the proof of Theorem~\ref{thm:undecmoni}.}
    \label{fig:undec-automatastrong}
\end{figure}

Figure~\ref{fig:undec-automatastrong} shows the TBA $\aut'$ with the locations of the (finite-word) timed automaton $\aut$ partitioned into $F$ and $Q\setminus F$.
In the figure, the new locations $\Dstate,\Estate,$ and $\Fstate$ and accompanying transitions are shown separately from the original automaton. Double circles denote accepting locations. Intuitively, processing a $\Bsymbol$ from an accepting location of $\aut$ takes the run to the accepting sink~$\Fstate$. 
On the other hand, processing a $\Bsymbol$ from a non-accepting location of $\aut$ takes the run to a component where the run continuation is accepting if and only if it processes infinitely many $\Bsymbol$'s.

It remains to show that the language $\lang(\aut')$ is strongly monitorable if and only if $\lang(\aut) = \finwords$.
One direction is trivial: If $\lang(\aut) = \finwords$, then $\lang(\aut') = \infwords_\Bsymbol$, as all words without a $\Bsymbol$ can be processed using the states of the complete automaton~$\aut$ (which are all accepting in $\aut'$) and all words with a $\Bsymbol$ can be accepted by moving to state~$\Fstate$ (which is an accepting sink) when processing the first $\Bsymbol$.  
Hence, $\lang(\aut')$ is strongly monitorable.

On the other hand, if $\lang(\aut) \neq \finwords$, then there exists a finite word~$\rho \notin \lang(\aut)$.
Hence, every run prefix of $\aut'$ processing $\rho \conc{\tau(\rho)} (\Bsymbol, \tau(\rho))$ must be in location~$\Dstate$.
Now, chose some $a \in \Sigma$.
Then, for all $\rho'' \in \finwords_\Bsymbol$, we have:
\begin{itemize}
    \item $\rho \conc{\tau(\rho)} (\Bsymbol, \tau(\rho)) \conc{} \rho'' \conc{} (a, 0)(a, 1)(a, 2)\cdots  \notin L(\aut')$ (as it contains only finitely many $\Bsymbol$). 
    \item $\rho \conc{\tau(\rho)} (\Bsymbol, \tau(\rho)) \conc{} \rho'' \conc{} (\Bsymbol, 0)(\Bsymbol, 1)(\Bsymbol, 2)\cdots  \in L(\aut')$ (as it contains infinitely many $\Bsymbol$).
\end{itemize}
Hence, $\monitoring{L(\aut')}(\rho \conc{\tau(\rho)} (\Bsymbol, \tau(\rho)) \conc{} \rho'') = \unknown$ for all $\rho''$.
Thus, $L(\aut')$ is not strongly monitorable.

%%%%%%%%%%%%%%%%%%%%%%%%%%%%%%%%%%
%%%%%%%%%%%%%%%%%%%%%%%%%%%%%%%%%% Weak part

In the case for weak monitorability we reduce the (undecidable~\cite{alur1994tba}) universality problem for TBA (not TA as we did in the strong case).
Additionally, as an intermediate step, we consider a problem about Brzozowski derivatives of properties of infinite timed words:
Let $\rho \in \finwords$ be a finite timed word and $\varphi\subseteq \infwords$ be such a property.
We say that $\rho$ is a universal prefix for $\varphi$ if the Brzozowski derivative 
    \[\{ \mu \in \infwords \mid \rho \conc{} \mu \in \varphi \}\]
    of $\varphi$ and $\rho$
is equal to $\infwords$.
Note that if $\varphi$ is equal to $\infwords$, then every prefix is universal for $\varphi$. 
However, the property~$(a,0)\conc{} \infwords$ has a universal prefix (e.g., $(a,0)$), but is not universal.

Further, we say that $\varphi$ is weakly $\top$-monitorable (cp.\ Definition~\ref{def:wbmonitorability}) if and only if there is a $\rho$ such that $\monitoring{\varphi}(\rho) = \top$, i.e., in comparison to (standard) weak monitorability, we only consider the verdict~$\top$.
Note that $\varphi$ has a universal prefix if and only if it is weakly $\top$-monitorable.

So, it remains to first reduce universality of TBA to the existence of a universal prefix for languages of TBA and then to reduce weak $\top$-monitorability to weak monitorability.

We begin with reducing TBA universality to the existence of universal prefixes. 
Intuitively, we add edges allowing to \myquot{restart} a run which allows every prefix to simulate the empty prefix.
Then, every prefix behaves like the empty prefix, which is universal if and only if the automaton is universal.

Let $\aut = (Q,Q_0,\Sigma,\clocks, \Delta, F)$ be a TBA.
We add a new letter~$\Bsymbol \notin \Sigma$ to obtain $\Sigma_\Bsymbol = \Sigma \cup \{\Bsymbol\}$ and construct a TBA~$\aut' = (Q',Q_0,\Sigma_\Bsymbol,\clocks, \Delta', F')$ from $\aut$ such that $\lang(\aut) = \infwords$ if and only if $\aut'$ has a universal prefix.
To this end, we introduce two new locations~$\Dstate$ and $\Estate$, i.e., $Q' = Q \cup \{\Dstate,\Estate\}$.
Next, we define $\Delta'$ (see Figure~\ref{fig:undec-automata_weak}) by copying the transitions from $\Delta$ and by adding the following transitions, where we write $q \xrightarrow{a} q'$ to denote $(q,q',a,\emptyset,\true)$: 
\begin{itemize}
    \item $\{ q \xrightarrow{a} \Dstate \mid q \in Q_0, a \in \Sigma_\Bsymbol\}$: from every initial location of $\aut$, there is an $a$-transition to $\Dstate$ for every $a \in \Sigma_\Bsymbol$.
    \item $\{ \Dstate \xrightarrow{a} \Estate, \Estate \xrightarrow{a} \Estate \mid a \in \Sigma \}$: for every $a \neq \Bsymbol$ there is an $a$-transition from $\Dstate$ to $\Estate$ and an $a$-labeled self-loop on $\Estate$.
    \item $\{ \Estate \xrightarrow{\Bsymbol} \Dstate, \Dstate \xrightarrow{\Bsymbol} \Dstate \}$: there is a $\Bsymbol$-transition from $\Estate$ to $\Dstate$ and a $\Bsymbol$-labeled self-loop on $\Dstate$.
    \item $\{ (q,q_0,\Bsymbol,\clocks,\true) \mid q \in Q, q \in Q_0\}$: from every location there is a $\Bsymbol$-transition to each initial location, which resets all clocks.
\end{itemize} 
We define the accepting locations of $\aut'$ as $F' = Q \cup \{\Dstate\}$.

\begin{figure}[ht]
    \centering

\begin{tikzpicture}[thick]
  % Draw the main wide rectangle with rounded corners
  \draw[thick,rounded corners] (0,0.2) rectangle (4,1.8);
  \node[anchor = north west ] at (0,1.8) {$\aut$};

    % Add automaton states

  % States on the left
  \node[state, inner sep =0,minimum size = 15] (l1) at (1,.6) {};
  \node[state, inner sep =0,minimum size = 15] (l2) at (2.5,.75) {};
  \node[state, inner sep =0,minimum size = 15] (l3) at (3.25,1) {};
  % \node[state, inner sep =0,minimum size = 15] (l4) at (4,.75) {};
  \node[state, accepting, inner sep =0,minimum size =15] (r) at (-1.5,.6) {$\Dstate$};
  \node[state, inner sep =0,minimum size =15] (s) at (-3.5,.6) {$\Estate$};

  \path[->, > = stealth]
  (0.4,1) edge (l1)
  (l1) edge[loop above] node[above] {$\Bsymbol$} (l1)
  (l2) edge[bend left] node[above] {$\Bsymbol$} (l1)
  (l3) edge[bend right] node[above] {$\Bsymbol$} (l1)
  % (l4) edge node[above] {$\Bsymbol$} (l1)
  (l1) edge node[above, near end] {$\Sigma_\Bsymbol$} (r)
  (r) edge[bend left] node[below] {$\Sigma$} (s)
  (s) edge[bend left] node[above] {$\Bsymbol$} (r)
  (r) edge[loop above] node[above] {$\Bsymbol$} ()
  (s) edge[loop above] node[above] {$\Sigma$} ()
  ;

\end{tikzpicture}
        \caption{The TBA $\aut'$ constructed for the proof of Theorem~\ref{thm:undecmoni}.}
    \label{fig:undec-automata_weak}
\end{figure}

Let $\lang(\aut) = \infwords$ and fix some $\mu' \in \infwords_\Bsymbol$.
If $\mu'$ contains infinitely many $\Bsymbol$'s, then it is accepted by $\autp$ using the new states $\Dstate$ and $\Estate$. On the other hand, if $\mu'$ contains only finitely many $\Bsymbol$'s, then it has the form 
\[
\mu' = 
\rho_0 \conc{} (\Bsymbol,t_0)
\conc{}
\rho_1 \conc{} (\Bsymbol,t_1)
\cdots 
 (\Bsymbol,t_{n-2})
 \conc{}
\rho_{n-1} \conc{} (\Bsymbol,t_{n-1})
\conc{}
\mu''
\]
for some $n \ge 0$, where the $\rho_i$ and $\mu''$ are $\Bsymbol$-free. 
As each $\rho_i$ is a prefix of some word in $\lang(\aut) = \infwords$ and $\mu''$ is in $\lang(\aut) = \infwords$, one can construct an accepting run of $\autp$ on $\mu'$.
Thus, $\lang(\autp) = \infwords_\Bsymbol$, i.e., $\varepsilon$ is a universal prefix of $\autp$.

Now for the other direction. If a word~$\rho$ is a universal prefix of $\autp$, then for all $\mu \in \infwords$, the word~$\rho \conc{} (\Bsymbol,0) \conc{} \mu$ is accepted by $\autp$. 
Hence, there is an accepting run of $\aut'$ on $\rho \conc{} (\Bsymbol,0) \conc{} \mu$. As all $\Bsymbol$-transitions lead to an initial state of $\aut$ and reset the clocks, and as $\mu$ does not contain any $\Bsymbol$'s, the suffix of the run on $\mu$ must also be an accepting run of $\aut$, i.e., we have $\mu \in \lang(\aut)$. Thus, $\lang(\aut) = \infwords$.
This concludes the first step of our proof.

In the second and last step of  our proof we reduce weak $\top$-monitorability to (standard) weak monitorability.
Intuitively, we manipulate the automaton so that it can never give the verdict~$\bot$ while preserving the existence of observations that yield the verdict~$\top$.

Let $\aut = (Q,Q_0,\Sigma,\clocks, \Delta, F)$ be a TBA.  
We add a new letter~$\Bsymbol \notin \Sigma$ to obtain $\Sigma_\Bsymbol = \Sigma \cup \{\Bsymbol\}$ and construct a TBA~$\autp = (Q',Q'_0,\Sigma_\Bsymbol,\clocks, \Delta', F')$ from $\aut$ such that $\aut$ is weakly $\top$-monitorable if and only if $\autp$ is weakly monitorable.

Now $Q'$ essentially contains two copies of $Q$: $Q_1=Q\times\{1\}$ and $Q_2=Q\times\{2\}$ as well as a new additional accepting location~$r$. The transition relation of $\autp$ is defined as follows (see also Figure~\ref{fig:reductionwb}), we write $q \xrightarrow{a} q'$ to denote $(q,q',a,\emptyset,\true)$:
\begin{itemize}

    \item $\{ (q,1) \xrightarrow{\Bsymbol} (q,1), (q,1) \xrightarrow{\Bsymbol} r \mid  q \in Q \}$: locations of $Q_1$ have a $\Bsymbol$-labeled self-loop and a $\Bsymbol$-transition to $r$.

    \item $\{ ((q,1),(q',2),a,\lambda, g) \mid (q,q',a,\lambda, g) \in \Delta \}$: locations of $Q_1$ have their $\Sigma$-labeled transitions redirected to $Q_2$. 
    
    \item $\{ (q,2),(q',2),a,\lambda, g) \mid (q,q',a,\lambda, g) \in \Delta \}$: locations of $Q_2$ have copies of the $\Sigma$-labeled transitions from $\aut$.
    
    % \item $\{ ((q,2),(q,1),\Bsymbol,\emptyset,\true) \mid q \in Q \}$: locations of $Q_2$ have $\Bsymbol$-labeled transitions directed back to $Q_1$.
    \item $\{ (q,2) \xrightarrow{\Bsymbol} (q,1) \mid q \in Q \}$: locations of $Q_2$ have $\Bsymbol$-labeled transitions directed back to $Q_1$.
    
    \item $\{ r \xrightarrow{\Bsymbol} r \}$: The new location $r$ has a $\Bsymbol$-labeled self-loop.

\end{itemize}
The set~$Q'_0$ of initial locations of $\autp$ is $Q_0\times\{1\}$ and the set~$F'$ of accepting locations of $\autp$ is $F\times\{2\}\cup \{r\}$.

\begin{figure}[h]
\centering

\begin{tikzpicture}[thick]
  % Draw the main wide rectangle with rounded corners
  \draw[thick,rounded corners] (0,0) rectangle (5,1.75);
  
  \draw[thick,rounded corners] (7,0) rectangle (12,1.75);

\node at (.35,.25) {$Q_1$};
\node at (7.35,.25) {$Q_2$};

\node[state, inner sep =0,minimum size = 15] (l1) at (1,.75) {};
\node[state, inner sep =0,minimum size = 15] (l2) at (2,.5) {};
\node[state, inner sep =0,minimum size = 15] (l3) at (3.25,1.25) {};
\node[state, inner sep =0,minimum size = 15] (l4) at (4,.75) {};

\node[state, inner sep =0,minimum size = 15] (r1) at (8,.75) {};
\node[state, inner sep =0,minimum size = 15,accepting] (r2) at (9,.5) {};
\node[state, inner sep =0,minimum size = 15] (r3) at (10.25,1.25) {};
\node[state, inner sep =0,minimum size = 15,accepting] (r4) at (11,.75) {};

\node[state, inner sep =0,minimum size = 15,accepting] (r) at (2.5,-.5) {$r$};

\node[minimum size = 30] (rr) at (11.5,.875){};

\path[->, > = stealth]
(l1) edge[bend right=50] node[below,near end] {$\Bsymbol$} (r)
(l2) edge[bend right] node[left, near end] {$\Bsymbol$} (r)
(l3) edge[bend left] node[yshift=-3,right,very near end] {$\Bsymbol$} (r)
(l4) edge[bend left=50] node[below,near end] {$\Bsymbol$} (r)

(l1) edge[loop above] node[left] {$\Bsymbol$} ()
(l2) edge[loop above] node[left] {$\Bsymbol$} ()
(l3) edge[loop left] node[left] {$\Bsymbol$} ()
(l4) edge[loop right] node[below] {$\Bsymbol$} ()

(r) edge[loop below] node[right] {$\Bsymbol$} ()

(5.1, 1.25) edge[bend left, ultra thick] node[above] {$\Sigma$} (6.9,1.25)
(6.9, .5) edge[bend left, ultra thick] node[below] {$\Bsymbol$} (5.1,.5)
(rr) edge[loop right, ultra thick] node[above,yshift=3] {$\Sigma$} ()
(-.25,.875) edge[ultra thick] (.35,.875);
;
\end{tikzpicture}
\caption{The TBA $\autp$ constructed for the proof of Theorem~\ref{thm:undecmoni}. Intuitively, $\autp$ has two disjoint copies of the locations of $\aut$ and $\Sigma$-labeled transitions lead from both copies to the second copy, while $\#$-labeled transitions from the second copy lead to the first copy. The first copy additionally has $\#$-labeled self-loops on all locations and $\#$-labeled transitions to $r$.}
\label{fig:reductionwb}
\end{figure}

Now, let $\rho\in T\Sigma^*$ be a timed word witnessing weak $\top$-monitorability of $\aut$, i.e., $\rho\cdot \mu\in L(\aut)$ for every $\mu\in \infwords$.  
We define $\rho'=\rho$.
In $\aut$, processing $\rho'$ leads to a location in $Q_2$. 
We argue that $\rho'$ witnesses weak monitorability of $\autp$. 
To this end, consider any $\mu'\in \infwords_\Bsymbol$ and let $\mu$ be obtained from $\mu'$ by removing all occurrences of $\Bsymbol$. 
If $\mu$ is infinite, we may simply mimic the accepting run of $\aut$ on $\mu$ from $\rho$ in $\autp$.
Clearly, this will also be accepting.
If $\mu$ is finite, $\mu'$ must have a suffix of the form~$(\Bsymbol^\omega,\tau)$, where $\tau$ is a sequence of timepoints.
This suffix is accepted by utilizing the $\Bsymbol$ transitions to $r$, from where the suffix~$(\Bsymbol^\omega,\tau)$ can be accepted.

For the opposite direction, let $\rho'\in T{\Sigma}_\Bsymbol^*$ be a timed word witnessing weak monitorability of $\autp$. 
That is, either $\rho'\cdot\mu'\in L(\autp)$ for all $\mu'$ or $\rho'\cdot\mu'\not\in L(\autp)$ for all $\mu'$. 
However, note that for any finite $\rho'$, we have $\rho'\cdot(\Bsymbol,0)(\Bsymbol,1)(\Bsymbol,2)\cdots  \in L(\autp)$, due to the $\Bsymbol$-labeled transition to $r$, from where $(\Bsymbol,0)(\Bsymbol,1)(\Bsymbol,2)\cdots $ is accepted. 
Thus, we must be in the first case where we have $\rho' \cdot\mu'\in L(\autp)$ for all $\mu'$.

Now, let $\rho\in T\Sigma^*$ be obtained from $\rho'$ by stripping all occurrences of $\Bsymbol$ in $\rho'$. 
We show that $\rho$ is a word witnessing weak $\top$-monitorability of $\aut$. 
Note that any run of $\autp$ on $\rho'$ reaching a location in $Q_1$ or $Q_2$ can be mimicked by a run of $\aut$ on $\rho$ reaching the corresponding location in $Q$. 
Let $\mu\in \infwords$. 
Then $\rho'\cdot\mu\in L(\autp)$. 
Clearly, this must be due to an accepting run where a suffix of the run processing $\mu$ stays in $Q_2$. 
Hence, $\rho\cdot\mu\in L(\aut)$, as we may mimic, for the prefix $\rho$, the prefix processing $\rho'$ of the accepting run of $\autp$ processing $\rho'\cdot\mu$, and, for the suffix $\mu$, we simply copy the run staying in $Q_2$.
\end{proof}

Due to Remark~\ref{rem_properties}, we also obtain the undecidability of strong and  weak $(\rho,t)$-moni\-torability.

\begin{corollary}
    Strong and weak $(\rho,t)$-monitorability is undecidable for properties given by TBA, even if $(\rho, t)$ is fixed.
\end{corollary}
%%%%%%%%%%%%%%%%%%%%%%%%%%%%%%%%%%%%%%%%%%%%%%%%
%%%%%%%%%%%%%%%%%%%%%%%%%%%%%%%%%%%%%%%%%%%%%%%%
%%%%%%%%%%%%%%%%%%%%%%%%%%%%%%%%%%%%%%%%%%%%%%%%

\section{Monitorability for DTMA}
\label{sec:dtma}

In this section, we show that (strong and weak) monitorability is decidable for properties given by DTMA.
The key difference to TBA, for which we have shown that monitorability is undecidable, is that they are trivially closed under complement~\cite{alur1994tba}, which we rely on in our proof.
This again demonstrates the importance of having automata for the property and its complement, as noted throughout this paper.

\newcommand{\acc}{\mathrm{Acc}}

Let us begin by introducing some notation and definitions.
Fix some TBA or TMA~$\aut = (Q, Q_0, \Sigma, \clocks, \Delta, \acc)$.
We write $(q_0, v_0) \xrightarrow{\rho}_\aut (q_n, v_n)$ for a finite timed word~$\rho = (\sigma, \tau) \in \finwords$ to denote the existence of a 
% \todo[color=green]{if we define finite runs we can just say that here} 
finite sequence \[(q_0, v_0) \xrightarrow{(\sigma_1, \tau_1)} (q_1,v_1) \xrightarrow{(\sigma_2, \tau_2)} \cdots \xrightarrow{(\sigma_n, \tau_n)} (q_n,v_n) \]
 of states,
where for all $1 \leq i \leq n$ there is a transition $(q_{i-1},q_{i},\sigma_{i},\lambda_i,g_i)$ such that $v_{i}(c) = 0$ for all $c$ in $\lambda_i$ and $v_{i-1}(c) + (\tau_i - \tau_{i-1})$ otherwise, and $g$ is satisfied by the valuation $v_{i-1}+(\tau_{i} - \tau_{i-1})$, where we use $\tau_0 = 0$. 

\begin{definition}[Non-Empty Language States, Reach Set, $\textit{Pre}^*$]
Let $\aut$ be a TBA or TMA.

\begin{itemize}
    
    \item The set of non-empty language states of $\aut$ is 
    \[\nestates{\aut} = \{ s \mid s \textit{ is a state of } \aut \textit{ and } \lang(\aut, s) \neq \emptyset \}.\]
Here, $\lang(\aut, s)$ is the set of infinite timed words accepted by a run starting in the state~$s$.

\item Given an observation~$(\rho, t)$, the set of states in which a run over $\rho$ can end starting from the initial states of $\aut$ after time~$t$ has passed is
\[\reachset\aut(\rho,t) = \bigcup\nolimits_{q_0 \in Q_0} \{(q, v + (t-\tau(\rho))) \mid (q_0, v_0) \xrightarrow{\rho}_\aut (q, v) \},\]
where $v_0$ is the clock valuation mapping every clock to $0$.
We call $\reachset\aut(\rho,t)$ the reach set of $(\rho, t)$ in \aut.

    \item    We define $\textit{Pre}^*_{\aut}(S)$ of a set~$S$ of states as
    \[\{(q,v) \mid (q,v) \xrightarrow{\rho}_\aut (q',v') \textit{ for some } (q',v' + t ) \in S, \rho \in \finwords \textit{, and } t \in \nnreals\}.\]
\end{itemize}
\end{definition}

A symbolic state is a pair~$(q, Z)$  of a location~$q$ and a zone~$Z$.
\label{def:zones}A zone is a finite conjunction of clock constraints of the form $x_1 \sim n$ or $x_1 - x_2 \sim n$, where $x_1, x_2$ are clocks, $\sim\, \in \{<, \le, =, \ge, >\}$, and $n \in \nnrats$. 
Such a zone describes a convex set of clock valuations.
Zones may be efficiently represented using so-called Difference-bounded Matrices (DBM)~\cite{DBLP:conf/ac/BengtssonY03}.

\begin{proposition}[\cite{DBLP:conf/formats/GrosenKLZ22}]
\label{lemma:algorithms}
There are zone-based algorithms for the following problems:
\begin{itemize}
    \item Given a TBA, compute its nonempty language states.
    \item Given a TBA or TMA~$\aut$ and an observation~$(\rho, t)$, compute the reach set~$\reachset{\aut}(\rho, t)$.
    \item Given a TBA or TMA~$\aut$ and a finite union~$S$ of symbolic states, compute $\textit{Pre}^*_{\aut}(S)$.
\end{itemize}
\end{proposition}

We continue by showing that the non-empty language states can also be computed for TMA, relying on an analysis of Alur and Dill's translation of TMA into equivalent TBA~\cite{alur1994tba}. 
Here, we just present the parts of their construction we need to prove our results.

\begin{lemma}
\label{lemma:tmaalgorithms}
There is a zone-based algorithm that, given a TMA, computes its nonempty language states.
\end{lemma}

%%%%%%%%%%%%%%%%%%%%%%%

\begin{proof}
Fix a TMA $\aut = (Q, Q_0, \Sigma, \clocks, \Delta, \mathcal{F})$. 
The translation from TMA to TBA relies on the fact that $\lang(\aut) = \bigcup_{F \in \mathcal{F}} \lang(\aut_F)$, where $\aut_F$ is the TMA~$(Q, Q_0, \Sigma, \clocks, \Delta, \{F\}\}$.
Alur and Dill translated each such $\aut_F$ into an equivalent TBA~$\aut_F'$ with set~$Q \times \{0,1,\ldots, |F|\}$ of locations and the same clocks as $\aut$. 
Then, the disjoint union of the $\aut_F'$ for $F \in \mathcal{F}$ is a TBA that is equivalent to $\aut$.

The TBA~$\aut_F'$ constructed by Alur and Dill satisfy the following fact: $\lang(\aut_F', ((q, 0), v)) = \lang(\aut_F, (q, v))$ for all locations~$q$ of $\aut$ and all clock valuations.
Thus, we have
\[\nestates{\aut} = \bigcup\nolimits_{F \in \mathcal{F}} \{(q, v) \mid ((q,0),v) \in \nestates{\aut_F'}\}.\]
Hence, we can symbolically compute the non-empty language states of $\aut$ by symbolically computing the non-empty states of the $\aut_F'$ and then project the states of the form~$(q,0)$ to $q$, but leaving the zones unchanged.
\end{proof}

%%%%%%%%%%%%%%%%%%%%%%%

Using this result, we can decide monitorability.

\begin{theorem}
Strong and weak $(\rho, t)$-monitorability are decidable for properties given by DTMA.
\end{theorem}

\begin{proof}
    We assume w.l.o.g.\ that the DTMA we are using are complete in the sense that every word has a run, which then must be unique due to determinism. 
    This can always be achieved by adding a fresh sink location and by rerouting all missing transitions to it (while preserving determinism).
    Then, the reach-set~$\reachset\aut(\rho, t)$ of any observation~$(\rho, t)$ is a singleton set, as $\aut$ is deterministic and complete.
    Thus, we identify $\reachset\aut(\rho, t)$ with the unique state in it.

    Now, given $\aut$, we first compute the set~$\nestates{\aut}$ of non-empty language states (see Lemma~\ref{lemma:tmaalgorithms}).
    Its complement is the set of empty language states, i.e., the set of states for which there is no accepting run starting there. Let us denote this set as~$\emptystates{\aut} = \comp{\nestates{\aut}}$ and note that we have $\reachset\aut(\rho,t) \in \emptystates{\aut} $ if and only if $ \monitoring{L(\aut)}(\rho,t) = \bot$, i.e., we can characterize the negative verdict in terms of the set~$\emptystates{\aut}$.

    We now use backwards reachability to compute the set of states that can reach $\emptystates{\aut}$. These are the states from where it is possible to reach a $\bot$-verdict. The sets \nestates{\aut}, \emptystates{\aut} and $\textit{Pre}^*_\aut(\emptystates{\aut})$ are illustrated in Fig.~\ref{fig:states}, where $\textit{Pre}^*_\aut(\emptystates{\aut})$ is the gray area.

    By definition, we have $\reachset\aut(\rho, t) \in \textit{Pre}^*_\aut(\emptystates{\aut})$ if and only if there exists a $\rho' \in \finwords$ and a $t' \in \nnreals$ such that $s_0 \xrightarrow{\rho\conc{t}\rho'} (q,v)$ with $(q,v+t') \in \emptystates{\aut}$, where $s_0$ denotes the unique initial state of $\aut$.
    The latter condition is equivalent to the existence of a $\rho' \in \finwords$ such that $s_0 \xrightarrow{\rho\conc{t}\rho'} s'$ for some $s' \in \emptystates{\aut}$.
    Hence, $\reachset\aut(\rho, t) \in \textit{Pre}^*_\aut(\emptystates{\aut})$ if and only if there exists a  $\rho' \in \finwords$  such that $\monitoring{L(\aut)}(\rho \conc{t} \rho') = \bot$.
    %for all $\mu \in \infwords$  we have $\rho \conc{t} \rho' \conc{} \mu \notin \lang(\aut)$. 
    Thus, $\reachset\aut(\rho, t) \notin \textit{Pre}^*_\aut(\emptystates{\aut})$ if and only if monitoring any extension of $(\rho, t)$ will not provide the verdict~$\bot$, i.e., $\monitoring{L(\aut)}(\rho \conc{t} \rho') \neq \bot$ for all $\rho' \in \finwords$.

    \begin{figure}[htb]
        \centering
        \begin{tikzpicture}
            % NE-states
            \draw[thick] (2,0) rectangle (6, 1.5);
            
            % Emptystates
            \draw[thick] (0,0) rectangle (2, 1.5);

            \node[circle, fill, inner sep=2pt] at (0.5, 0.4) {};
            \node[circle, fill, inner sep=2pt] at (1.2, 1.25) {};
            \node[circle, fill, inner sep=2pt] at (2.7, 1.3) {};
            \node[circle, fill, inner sep=2pt] at (2.5, 0.2) {};
            \node[circle, fill, inner sep=2pt] at (3.5, 0.3) {};
            
            \node[circle, fill, inner sep=2pt] at (4.7, 0.5) {};
            \node[circle, fill, inner sep=2pt] at (5.4, 1.27) {};
            \node[circle, fill, inner sep=2pt] at (5, 1.2) {};
            \node[circle, fill, inner sep=2pt] at (5.7, 0.8) {};
     
            \node[circle, fill, inner sep=2pt] (s1) at (1.4, 0.8) {};
            \node[circle, fill, inner sep=2pt] (s2) at (3, 0.6) {};
            \node[circle, fill, inner sep=2pt] (s3) at (4.3, 0.8) {};

            \draw[->, > = stealth, very thick, dashed] (s1) edge [bend left] (s2);
            \draw[->, > = stealth, very thick] (s2) edge [bend left] (s1);

            \node[fill=white, rectangle, anchor=north] at (4, 1.5) {$\nestates{\aut}$};
            \node[fill=white, rectangle, anchor=north west] at (0, 1.5) {$\emptystates{\aut}$};
            % \node[fill=white, rectangle, anchor=north] at (2, 2) {\(\pestates{\aut}\)};

            \draw[very thick] (0,0) rectangle (6, 1.5);
            
            % \node at (-1, 1) {\(\aut\)};
  
        %    \draw[decorate, decoration={brace, amplitude=10pt}, thick,white] (4,-.1) -- (0,-.1) node[midway, below=10pt] {\(\pestates{\aut}\)};
        \end{tikzpicture}
        \hfill
        \begin{tikzpicture}    
            % Possibly empty states
            \draw[thick, fill=gray!20] (0,0) -- (4, 0) -- (3, 1.5) -- (0, 1.5) -- cycle;
            % NE-states
            \draw[thick] (2,0) rectangle (6, 1.5);
            
            % Emptystates
            \draw[thick] (0,0) rectangle (2, 1.5);

            \node[circle, fill, inner sep=2pt] at (0.5, 0.4) {};
            \node[circle, fill, inner sep=2pt] at (1.2, 1.25) {};
            \node[circle, fill, inner sep=2pt] at (2.7, 1.3) {};
            \node[circle, fill, inner sep=2pt] at (2.5, 0.2) {};
            \node[circle, fill, inner sep=2pt] at (3.5, 0.3) {};
            
            \node[circle, fill, inner sep=2pt] at (4.7, 0.5) {};
            \node[circle, fill, inner sep=2pt] at (5.4, 1.27) {};
            \node[circle, fill, inner sep=2pt] at (5, 1.2) {};
            \node[circle, fill, inner sep=2pt] at (5.7, 0.8) {};

            \node[circle, fill, inner sep=2pt] (s1) at (1.4, 0.8) {};
            \node[circle, fill, inner sep=2pt] (s2) at (3, 0.6) {};
            \node[circle, fill, inner sep=2pt] (s3) at (4.3, 0.8) {};

            \draw[->, > = stealth, very thick, dashed] (s3) edge [bend left] (s2);
            \draw[->, > = stealth, very thick] (s2) edge [bend left] (s3);

            \draw[->, > = stealth, very thick, dashed] (s1) edge [bend left] (s2);
            \draw[->, > = stealth, very thick] (s2) edge [bend left] (s1);
            
            \node[ rectangle, anchor=north] at (4, 1.5) {$\nestates{\aut}$};
            \node[ rectangle, anchor=north west] at (0, 1.5) {$\emptystates{\aut}$};
            %\node[ rectangle, anchor=south] at (2, 2) {\(\pestates{\aut}\)};

            \draw[very thick] (0,0) rectangle (6, 1.5);
            
            % \node at (-1, 1) {\(\aut\)};
            
        \end{tikzpicture}
        \caption{
        Representation of all states of an automaton \aut. On the left, only the empty language states \emptystates{\aut} and non-empty language states \nestates{\aut} are shown. On the right, the states in $\textit{Pre}^*_\aut(\emptystates{\aut})$ are gray. 
        The solid arrows are examples of possible transitions and the dashed arrows are examples of impossible transitions.
        A state in \emptystates{\aut} has no accepting run, thus cannot reach \nestates{\aut}. A state in \nestates{\aut} might reach a state in \emptystates{\aut}. A state outside $\textit{Pre}^*_\aut(\emptystates{\aut})$ cannot reach~\emptystates{\aut}.
        }
        \label{fig:states}
    \end{figure}

    Since DTMA are complementable, we can compute $\emptystates{\comp\aut}$ for a complement automaton~$\comp\aut$ of $\aut$ and provide the dual characterization for the $\top$-verdict: $\reachset{\comp\aut}(\rho, t) \notin \textit{Pre}^*_{\comp\aut}(\emptystates{\comp\aut})$ if and only if $\monitoring{L(\aut)}(\rho \conc{t} \rho') \neq \top$ for all $\rho' \in \finwords$.
    Here, we rely on the fact that we have $\monitoring{L}(\rho, t) = \top$ if and only if $\monitoring{\comp{L}}(\rho, t) = \bot$ and $\monitoring{L}(\rho, t) = \bot$ if and only if $\monitoring{\comp{L}}(\rho, t) = \top$.

    Now, recall that we can complement DTMA by only changing the acceptance condition, i.e., we can assume w.l.o.g.\ that $\aut$ and $\comp\aut$ have the same locations, clocks, and transitions, and  thus we have $\reachset\aut(\rho, t) = \reachset{\comp\aut}(\rho,t)$ for all $(\rho, t)$.

    Hence, using the above characterizations, we have 
$\reachset\aut(\rho, t) \in \emptystates{\aut} \cup \emptystates{\comp\aut}$  if and only if $ \monitoring{L(\aut)}(\rho,t)\in \{\top, \bot\}$
and thus
$ \reachset\aut(\rho, t) \in \textit{Pre}^*_\aut(\emptystates{\aut} \cup \emptystates{\comp\aut})$ if and only if
there exists a $\rho' \in \finwords$ such that $\monitoring{L(\aut)}(\rho \conc{t} \rho')\in \{\top, \bot\}$.
Thus, $L(\aut)$ is weakly $(\rho,t)$-monitorable if and only if
    \[\reachset\aut(\rho,t) \in \textit{Pre}^*_\aut(\emptystates{\aut} \cup \emptystates{\comp\aut}).\]

    For strong monitorability, we need one extra step corresponding to the additional quantifier in the definition.
    By complementing $\textit{Pre}^*_\aut(\emptystates{\aut} \cup \emptystates{\comp\aut})$, we obtain the set of states for which the monitoring verdict will forever be $\unknown$: we have 
    $\reachset\aut(\rho,t) \in \comp{\textit{Pre}^*_\aut(\emptystates{\aut} \cup \emptystates{\comp\aut})}$ if and only if 
    $\monitoring{L(\aut)}(\rho \conc{t} \rho') = \unknown$ for all $\rho' \in \finwords$.

    Then, by using backwards reachability another time, we can compute the set of states from where it is possible to reach a state from which only the verdict~$\unknown$ can be given: we have $\reachset\aut(\rho,t ) \in \textit{Pre}^*_\aut(\comp{\textit{Pre}^*_\aut(\emptystates{\aut} \cup \emptystates{\comp\aut})})$ if and only if  there exists a $\rho' \in \finwords$ such that for all  $\rho'' \in \finwords$, we have $\monitoring{L(\aut)}(\rho\conc{t}\rho'\conc{}\rho'')=\unknown$.
    Thus, $L(\aut)$ is strongly $(\rho,t)$-monitorable if and only if \[\reachset\aut(\rho,t) \notin \textit{Pre}^*_\aut(\comp{\textit{Pre}^*_\aut(\emptystates{\aut} \cup \emptystates{\comp\aut})}).\]  

    Due to Lemma~\ref{lemma:algorithms} and Lemma~\ref{lemma:tmaalgorithms}, both the characterization of weak $(\rho,t)$-monitorability and the one of strong $(\rho,t)$-monitorability are effectively decidable.
\end{proof}

Due to Remark~\ref{rem_properties}, we also obtain the decidability of strong and weak monitorability.

\begin{corollary}
    Strong and weak monitorability are decidable for properties given by DTMA.
\end{corollary}

%%%%%%%%%%%%%%%%%%%%%%%%%%%%%%%%%%%%%%%%%%%%%%%%
%%%%%%%%%%%%%%%%%%%%%%%%%%%%%%%%%%%%%%%%%%%%%%%%
%%%%%%%%%%%%%%%%%%%%%%%%%%%%%%%%%%%%%%%%%%%%%%%%
\section{Monitorability with Step-Bounded Horizons}
\label{sec:bmtba}
Strong monitorability of a property ensures that at any time during monitoring, the observation $(\rho,t)$ made so far can be extended by some finite $\rho'$ such that a conclusive verdict can be made after $\rho \conc{t} \rho'$.  Thus, in the setting of strong monitorability, it is always meaningful to keep monitoring.
In contrast, for weakly monitorable properties, the ability to make a conclusive verdict after some future (finite) monitoring is not guaranteed.  Ideally, we would like to refine the rather uninformative verdict~$\unknown$ with guaranteed minimum bounds on the number of future events before a positive or negative verdict can be made. In case both these bounds are $\infty$, further monitoring is useless as the verdicts reported will always be $\unknown$.
Here it is prudent to distinguish between the two definitive verdicts, as they have different decidability properties.

\begin{definition}[Weak Monitorability with Step-Bounded Horizons]\label{def:wbmonitorability}
Fix a property~$\varphi \subseteq \infwords$, an observation~$(\rho, t)$, and a bound $n\in\nats$. 
\begin{itemize}

    \item $\varphi$ is bounded weakly $\top$-$(\rho, t)$-monitorable with respect to $n$ if and only if
\[
    \text{ there exists } \rho' \in T\Sigma^{\leq n} \text{ such that } \monitoring{\varphi}(\rho\conc{t}\rho') = \top.
  \]

  \item $\varphi$ is bounded weakly $\bot$-$(\rho, t)$-monitorable with respect to $n$ if and only if
\[
    \text{ there exists }  \rho' \in T\Sigma^{\leq n} \text{ such that } \monitoring{\varphi}(\rho\conc{t}\rho') =\bot.
  \]

  \item $\varphi$ is bounded weakly $(\rho, t)$-monitorable with respect to $n$ if and only if
\[
    \text{ there exists }  \rho' \in T\Sigma^{\leq n} \text{ such that } \monitoring{\varphi}(\rho\conc{t}\rho') \in \{\top,\bot\}.
  \]

\item $\varphi$ is bounded weakly $\top$-monitorable with respect to $n$ if it is bounded weakly $\top$-$(\varepsilon, 0)$-monitorable with respect to $n$.

\item $\varphi$ is bounded weakly $\bot$-monitorable with respect to $n$ if it is bounded weakly $\bot$-$(\varepsilon, 0)$-monitorable with respect to $n$.

\item $\varphi$ is bounded weakly monitorable with respect to $n$ if it is bounded weakly $(\varepsilon, 0)$-monitorable with respect to $n$.

\end{itemize}
\end{definition}

In the following, we show that bounded weak monitorability and bounded weak $\top$-monitorability are undecidable, but that bounded weak $\bot$-monitorability is decidable. 

\begin{theorem}
\label{thm:topundec}
Bounded weak $\top$-monitorability
is undecidable for properties given by TBA.
\end{theorem}

\begin{proof}
Let  $\varphi$ be a property given by a TBA $\aut$. Then universality of $\aut$ is equivalent to bounded weak $\top$-monitorability of $\varphi$ for the bound $n=0$. As universality for timed automata is undecidable~\cite{alur1994tba}, $0$-bounded weak $\top$-monitorability for TBA properties is undecidable.
\end{proof}

Next, we show that weak $\bot$-$(\rho,t)$-monitorability behaves differently. 

\begin{theorem}
\label{thm:botdec}
Bounded weak $\bot$-$(\rho, t)$-monitorability
is decidable for properties given by TBA.
\end{theorem}

\begin{proof}
Let $(\rho, t)$ be an observation with $\rho = (\sigma_1,t_1)(\sigma_2, t_2)\cdots  (\sigma_m,t_m)$ and let the property be given by the TBA~$\aut$.

Let $\pi = q_0\xrightarrow{\sigma_1,\lambda_1,g_1}q_1\cdots q_{m+n-1}\xrightarrow{\sigma_{m+n},\lambda_{m+n},g_{m+n}}q_{m+n}$ be a syntactic path of length~$m+n$ of the TBA $\aut$ induced by transitions~$(q_i, q_{i+1},\sigma_i, \lambda_i, g_i)$. 
Also, let $\tau=\{\tau_1,\ldots,\tau_{m+n}\}$ be the set of variables ranging over global time (i.e., $\nnreals$), where $\tau_i$ denotes the time at which the $i$-th transition is taken.  Clearly, $\tau_i=t_i$ for $i=1,\ldots, m$.
Most importantly, there is a zone~$Z_\pi(\tau)$ over $\tau$ that precisely captures when $(\sigma,\tau)$ is a timed word realizing $\pi$. 
In this case, we may also express the resulting clock valuation after realizing $(\sigma,\tau)$ on $\pi$ as $v_\pi(\tau)=(v_\pi^1, \ldots,v_\pi^k)$, where $k$ is the number of clocks of $\aut$ and $v_\pi^i=\tau_{m+n}-\tau_{\ell(i)}$, with $\ell(i)$ being the index of the last transition when the clock $x_i$ was reset. 
Now, if $(q_{m+n},v_\pi(\tau))\not\in\nestates{\aut}$, then the run of $(\sigma,\tau)$ following $\pi$ cannot be extended to an accepting run of $\aut$.  

Witnesses for weak monitorability are closed under extensions.
Hence, in the following construction, we can restrict ourselves w.l.o.g.\ to words of length exactly~$n$ instead of words of length at most~$n$.
The following formula expresses the existence of a timed word~$\rho' = (\sigma',\tau')$ of length~$n$ such that all runs on $\rho \conc{t} \rho'$ end in a state outside of $\nestates{\aut}$, witnessing $\monitoring{\lang(\aut)}(\rho\conc{t}\rho')=\bot$:
\[
\exists \tau_1\ldots\exists\tau_{m+n}.\, 
\bigwedge_{i=1}^m \tau_i = t_i \wedge \tau_m \le t \le \tau_{m+1} \wedge \left(\bigvee_{\sigma_{m+1}\cdots\sigma_{m+n} \in \Sigma^n}
\bigwedge_{\pi \in \mathrm{SP}}
 \left( Z_\pi(\tau) 
\rightarrow v_\pi(\tau)\not\in\nestates{\aut }\right)\right),
\]
where $\mathrm{SP}$ is the set of syntactic paths of the form 
\[\pi=q_0\xrightarrow{g_1,\sigma_1,\lambda_1}q_1\cdots\xrightarrow{g_{m},\sigma_{m},\lambda_{m}}q_m\cdots q_{m+n-1}\xrightarrow{\sigma_{m+n},\lambda_{m+n},g_{m+n}}q_{m+n}.\]
The formula above expresses the existence of values for  $\tau_1,\ldots,\tau_{m+n}$, with $\tau_i=t_i$ for $i=1,\ldots, m$, and $\tau_m\leq t \leq \tau_{m+1}$, ensuring that only extensions of the observation $(\rho,t)$ are considered. The last part of the formula ensures the existence of a timed word $\rho'=(\sigma',\tau')$, where $\sigma'=\sigma_{m+1}\cdots\sigma_{m+n}$ and $\tau'=\tau_{m+1}\cdots\tau_{m+n}$, where all runs -- which are captured by $Z_\pi(\tau)$ -- are outside $\nestates{\aut}$.
The formula is in the first order theory of real-closed fields, which is decidable~\cite{tarski}. Thus, bounded weak $\bot$-monitorability is decidable for properties given by TBA.
\end{proof}

\begin{remark}
It follows from Theorem~\ref{thm:botdec} that for properties where both the property and its complement are given by TBA, bounded weak monitorability is decidable.
Hence, if the property is given by a DTMA, then bounded weak monitorability is decidable. Moreover, one can even compute the tightest bound~$n$ by doing a breadth-first search over the symbolic state graph of the deterministic timed automaton, searching for a (shortest) path to the empty language states.
% standard reachability techniques (i.e., breadth-first zone-based search) for timed automata.
\end{remark} 

Lastly, we prove bounded weak monitorability undecidable for properties given by TBA.

\begin{theorem}
\label{thm:undecweakbounded}
Bounded weak monitorability
is undecidable for properties given by TBA.
\end{theorem}

\begin{proof}
In the second step of the undecidability proof of weak monitorability (see Theorem~\ref{thm:undecmoni}), we have shown how to reduce weak $\top$-monitorability to weak monitorability. 
The same construction reduces bounded weak $\top$-monitorability to bounded weak monitorability, as the length of the witnesses is only decreased by removing $\Bsymbol$'s.
\end{proof}

Due to Remark~\ref{rem_properties}, we also obtain results for the remaining cases.

\begin{corollary}
Bounded weak $\top$-$(\rho, t)$-monitorability and bounded weak $(\rho, t)$-monitorability are undecidable (even for fixed $(\rho, t)$) while bounded weak $\bot$-monitorability is decidable.
\end{corollary}

%%%%%%%%%%%%%%%%%%%%%%%%%%%%%%%%%%%%%%%%%%%%%%%%
%%%%%%%%%%%%%%%%%%%%%%%%%%%%%%%%%%%%%%%%%%%%%%%%
%%%%%%%%%%%%%%%%%%%%%%%%%%%%%%%%%%%%%%%%%%%%%%%%
\section{Refined Monitoring with Time-Horizon Verdicts}
\label{sec:prediction}

In the previous section, we have shown that the uninformative verdict~$\unknown$ can be refined by checking whether within a bounded number of new observations, a definitive verdict may be given.
In this section, we again refine the uninformative verdict~$\unknown$ by computing lower bounds on the time that needs to pass (independently of the number of events observed during that time) before a definitive verdict may be given. 
Again, if this is infinite, then the monitoring process can be stopped, as no amount of waiting will yield a definitive verdict.
This refinement was introduced and briefly studied as ``time predictive monitoring'' by Grosen et al.~\cite{DBLP:conf/formats/GrosenKLZ22}. 
Here we revisit (and rename) this notion by proving computability of the time-bounded monitorability for properties given by a DTMA.

Recall that time passing without any new observed events can nevertheless yield definitive verdicts (see, e.g., the last two items in Example~\ref{example_monitoring}).
Thus, in a practical setting, one is interested in intermittently querying the monitoring function even if no events are observed. 
Here, we give lower bounds on the time one should let pass before the next such query is made, thereby reducing the computational overhead of these queries.

\begin{definition}[Refined Monitoring with Time-Horizon Verdicts]
\label{def_refinedmonitoring}
Given an observation~$(\rho,t)$ and a property~$\varphi \subseteq \infwords$, the refined monitoring function $\predmonitoring{\varphi}$ is defined as 
\[
\predmonitoring{\varphi}(\rho,t) = 
\begin{cases}
\top & \text{if } \monitoring{\varphi}(\rho,t)=\top,\\
\bot & \text{if } \monitoring{\varphi}(\rho,t)=\bot,\\
( \pospredmonitoring{\varphi}(\rho,t), \negpredmonitoring{\varphi}(\rho,t)) & \text{otherwise} ,\\
\end{cases}
\] with
\begin{itemize}
    \item $\pospredmonitoring{\varphi}(\rho,t) = \inf\big\{ t' \mid \rho' \in \finwords \text{ such that } \monitoring{\varphi}(\rho\conc{t} \rho', t') = \top \text{ and } t' \ge \tau(\rho\conc{t} \rho')  \big\}$ and
    \item $\negpredmonitoring{\varphi}(\rho,t) = \inf\big\{ t' \mid \rho' \in \finwords \text{ such that } \monitoring{\varphi}(\rho\conc{t} \rho', t') = \bot \text{ and } t' \ge \tau(\rho\conc{t} \rho') \big\},$
\end{itemize}
where we use the convention~$\inf \emptyset = \infty$.
\end{definition}

\begin{example}
Consider the MITL property~$\varphi = F_{[20,40]} b$. Monitoring the finite timed word~$\rho=(a,5.1) (c,21.0) (c, 30.4) (b, 35.1) (a, 40.2)$ will result in three $\unknown$ verdicts followed by the verdict $\top$ when $(b,35.1)$ is read.  However,  we may offer significantly more information, e.g., when reading $(a, 5.1)$ it is clear that at least $14.9$ time-units must elapse before we can give the verdict~$\top$, and at least $34.9$ time-units must elapse before we can give the verdict~$\bot$.  
Hence, $\predmonitoring{\varphi} ((a, 5.1),5.1) = (14.9,34.9)$.
\end{example}

\begin{remark}
$\monitoring{\varphi}(\rho, t) = \top$ implies $\pospredmonitoring{\varphi}(\rho,t) = 0$ and $\negpredmonitoring{\varphi}(\rho,t) = \infty$ and $\monitoring{\varphi}(\rho, t) = \bot$ implies $\pospredmonitoring{\varphi}(\rho,t) = \infty$ and $\negpredmonitoring{\varphi}(\rho,t) = 0$, but the converse is in general not true. 
Consider, for example, the MITL property~$\varphi = F_{\geq 0} a$ and the observation~$\rho=(b,1)$. 
Then, $\pospredmonitoring{\varphi }(\rho,1)=0 $ (witnessed by $\rho' = (a,0)$ for which we have $\monitoring{\rho \conc{t} \rho'} = \top$) and $\negpredmonitoring{\varphi} = \infty$, but $\monitoring{\varphi}(\rho,1) = \unknown$.
\end{remark}

Time-bounded monitorability for an observation~$(\rho, t)$ refines weak $(\rho, t)$-monitorability.

\begin{lemma}
\label{connection}
Let $(\rho, t)$ be an observation and $\varphi$ a property. Then, $\varphi$ is weakly $(\rho,t)$-monitorable if and only if at least one of the values~$\pospredmonitoring{\varphi}(\rho,t) $ and $\negpredmonitoring{\varphi}(\rho,t) $ is finite.
\end{lemma}

\begin{proof}
Let $\varphi$ be weakly $(\rho,t)$-monitorable, i.e., there is a $\rho' \in \finwords$ such that $\monitoring{\varphi}(\rho\conc{t}\rho') \in \{\top, \bot\}$, say it is $\top$ (the other case is analogous).
Then, $\rho'$ witnesses $\pospredmonitoring{\varphi}(\rho,t) \le \tau(\rho')$, i.e., $\pospredmonitoring{\varphi}(\rho,t)$ is finite.
On the other hand, if (say) $\pospredmonitoring{\varphi}(\rho,t)$ is finite (the other case is analogous), then there is a $\rho'$ such that $\monitoring{\varphi}(\rho\conc{t} \rho') = \top$.
Hence, $\varphi$ is weakly $(\rho,t)$-monitorable.
\end{proof}

\begin{theorem}[Refined Monitoring is Effective]
\label{thm_refined}
$\predmonitoring{\varphi}$ is effectively computable, if $\varphi$ is given by a DTMA.
\end{theorem}

\begin{proof}
When online-monitoring $\varphi$ over some observation~$(\rho , t)$, we compute the reach-set~$\reachset{\aut}(\rho, t)$ and check if it is a subset of the non-empty language states~$\nestates{\aut}$ or a subset of the empty language states~$\emptystates{\aut}$~\cite{DBLP:conf/formats/GrosenKLZ22}. 
Thus, $\negpredmonitoring{\varphi}(\rho,t)$ is the infimum of the time duration of all paths from $\reachset{\aut}(\rho,t)$ to $\emptystates{\aut}$.
This is a time-optimal reachability problem.

Asarin and Maler~\cite{DBLP:conf/hybrid/AsarinM99} introduced the concept of time-optimal strategies for timed game automata.
Since timed game automata trivially generalize timed automata, we can adapt the computation of the optimal time bound to obtain the minimal possible time to reach an empty language state from the reach-set. This gives us the value~$\negpredmonitoring{\aut}(\rho,t)$.

Since DTMA are closed under complement, the same procedure for the complement automaton gives us the value~$\pospredmonitoring{\aut}(\rho,t)$.
\end{proof}

As explained in the proof of Theorem~\ref{thm_refined}, monitoring is implemented by keeping track of the reach-set of the observation and checking at each update whether it is contained in the non-empty language states or in the empty language states.
But, as explained in the introduction of this section, one should not only update the reach-set when a new event is observed, but also intermittently when time has passed. 

% If we only consider time passing. We can compute the bound of how much time can pass, before the monitoring verdict changes from $\unknown$ to $\top$ or $\bot$.

For this special case where only time passes, which is covered in Definition~\ref{def_refinedmonitoring} by considering $\rho' = \varepsilon$ in the infimum, 
we do not need to solve the expensive time-optimal reachability problem described in the proof of Theorem~\ref{thm_refined}. Instead, we rely on zone operations to compute 
$\delta_\varphi(\rho,t) = \inf\{d \mid \monitoring{\varphi}(\rho,t+d) \in \{\top, \bot\}\}$. 
This is done by exploring all delays (removing the upper bounds of the zones) and subtracting the non-empty language states. The lower bounds in the zones of the resulting states gives the minimum time a verdict can be made by waiting.
Now, having observed $\rho$ and the current time being $t$ such that $\monitoring{\varphi}(\rho, t) = \unknown$, querying the monitoring function again without a new observation before time $t + \delta_\varphi(\rho,t) $ will not yield a different verdict and can thus be avoided.

%%%%%%%%%%%%%%%%%%%%%%%%%%%%%%%%%%%%%%%%%%%%%%%%
%%%%%%%%%%%%%%%%%%%%%%%%%%%%%%%%%%%%%%%%%%%%%%%%
%%%%%%%%%%%%%%%%%%%%%%%%%%%%%%%%%%%%%%%%%%%%%%%%
 \section{Related Work}
 \label{sec:related}
A formal notion of monitorability was first introduced by Pnueli and Zaks in their work on monitoring Property Specification Language (PSL)~\cite{pnueli2006psl}.
In that work, the authors defined strong monitorability given a finite prefix, called $\sigma$-monitorability, on which we base our $(\rho,t)$-monitorability.
From this, Bauer, Leuker, and Schallhart defined the most common definition of monitorability, that is strong $\sigma$-monitorability from the initial state~\cite{bauer2011runtime}.
They proved that safety and guarantee properties are a proper subset of the class of strongly monitorable properties.
Later, Chen et al.~\cite{chen2018deciding} and Peled and Havelund~\cite{peled2019refining} noticed that there exist properties that are not strongly monitorable but still have utility to monitor, and proposed equivalent definitions of weak monitorability.
Mascle \etal showed the monitorability of LTL properties can be improved by considering robust semantics~\cite{rltl}.
The term \emph{strong monitorability} has been used before in the context of partially observable stochastic systems modeled as Hidden Markov Models.
Sistla, \v{Z}efran, and Feng first used the term, contrasted with (standard) monitorability~\cite{sistla2011monitorability} in their work extending the results of Gondi, Patel, and Sistla on monitoring $\omega$-regular properties of stochastic systems~\cite{gondi2009monitoring}.

The complexity of monitorability problems has been studied in other untimed settings.
Diekert and Leuker proposed a topological definition of strong monitorability, showing that problem is equivalent to showing that the boundary in the Cantor topology has an empty interior~\cite{diekert2014topology}.
Diekert, Muscholl, and Walukiewicz later proved that deciding monitorability for (untimed) Büchi automata is \pspace-complete~\cite{diekert2015note}.
Agrawal and Bonakdarpour proposed a definition of monitorability for hyperproperties and determined the monitorable classes for their three-valued specification language, HyperLTL~\cite{clarkson2014hyperltl}.
Francalanza, Aceto, and Ingolfsdottir characterized monitorable properties of the branching-time $\mu$-Hennessy-Milner Logic~\cite{francalanza2017monitorability}.
This was later extended by Aceto \etal who introduced a hierarchy of monitorable fragments of the language~\cite{aceto2019adventures}.

Attempts have been made to unify these different notions of monitorability.
Peled and Havelund introduced a classification for properties centered around monitorability~\cite{peled2019refining}.
Kauffman, Havelund, and Fischmeister defined a common notation for strong and weak monitorability for different verdict domains~\cite{kauffman2021what}.
Aceto, Achilleos, and Francalanza provided syntactic characterizations of monitorability for classical notions of monitorability as well as for a variant of the modal $\mu$-calculus, recHML~\cite{aceto2021operational}.

We are aware of only one other work addressing monitorability for real-time properties.
Amara~\etal (very recently) introduced a new linear-time timed $\mu$-calculus, which subsumes MTL, and therefore also MITL, and identified its largest monitorable fragment~\cite{amara:hal-05043055}.
Their work differs from ours in three important respects.
First, they rely on a definition of monitorability introduced by Schneider as \emph{Execution Monitoring Enforceability}~\cite{schneider2000enfoceable} while we use the more typical definition of strong and weak monitorability due to Pnueli and Zaks~\cite{pnueli2006psl}.
Second, they introduce a new calculus and characterize its maximal monitorable fragment.
We instead consider the case of properties expressed as Timed Automata, which lend themselves to algorithmic manipulation.
Finally, we prove decidability and undecidability results. 
%%%%%%%%%%%%%%%%%%%%%%%%%%%%%%%%%%%%%%%%%%%%%%%%
%%%%%%%%%%%%%%%%%%%%%%%%%%%%%%%%%%%%%%%%%%%%%%%%
%%%%%%%%%%%%%%%%%%%%%%%%%%%%%%%%%%%%%%%%%%%%%%%%
\section{Conclusion}
\label{sec:conclusion}

In this work, we have studied monitorability for timed properties specified by either (possibly nondeterministic) TBA or by deterministic TMA.
In general, we proved  monitorability decidable for specifications given by deterministic automata and undecidable for specifications given by nondeterministic automata.
The notable exception here is bounded weak $\bot$-monitorability, which is even decidable for nondeterministic TBA.

Also, we provided refinements of monitoring and monitorability making the verdict~$\unknown$ more informative by providing bounds on the number of events or the amount of time that needs to pass before a conclusive verdict may be given.
In practical settings, this is crucial information and also allows to optimize the monitoring process in the real-time setting.

Our decidability proof for monitorability of DTMA relies on the fact that DTMA can be complemented without changing the state space.
On the other hand, monitorability is undecidable if the property is given by a TBA. 
Thus, another question for further research is to consider strong monitorability when given TBA for the property \emph{and} for its negation.
This is a very natural setting, as the specification logic~MITL is closed under negation and can be translated into TBA. 
Also, monitoring often requires TBA for both the property and its negation~\cite{DBLP:conf/formats/GrosenKLZ22, delay,assumptions}. 

\textbf{Acknowledgements.} We would like to thank Corto Mascle, who brought Rampersad et al.'s work on the complexity of suffix-universality~\cite{suffix} to our attention, which inspired our undecidability proof for weak monitorability.

\bibliographystyle{plain}
\bibliography{bib}

\end{document}

%% file: macros.tex
\newcommand{\nnreals}{\ensuremath{\mathbb{R}_{\ge 0}}\xspace}
\newcommand{\nats}{\ensuremath{\mathbb{N}}\xspace}
\newcommand{\nnrats}{\ensuremath{\mathbb{Q}}\xspace}

\newcommand{\infwords}{\ensuremath{T\Sigma^\omega}\xspace}
\newcommand{\finwords}{\ensuremath{T\Sigma^*}\xspace}
\newcommand{\emptyword}{\ensuremath{\varepsilon}\xspace}
\newcommand{\conc}[1]{\ensuremath{\cdot_{#1}}\xspace}

\newcommand{\monitoring}[1]{\ensuremath{V_{#1}}\xspace}
\newcommand{\lang}{\textit{L}}

\newcommand{\true}{\texttt{true}}
\newcommand{\false}{\texttt{false}}

\newcommand{\aut}{\ensuremath{\mathcal{A}}\xspace}
\newcommand{\autp}{\ensuremath{\aut^{\prime}}\xspace}
\newcommand{\clocks}{\ensuremath{\mathcal{X}}\xspace}

\newcommand{\nestates}[1]{\ensuremath{S^{ne}_{#1}}\xspace}

\newcommand{\emptystates}[1]{\ensuremath{S^{\emptyset}_{#1}}\xspace}

\newcommand{\etal}[0]{et al.~}

\newcommand{\pspace}{\textsc{PSpace}\xspace}

\newcommand{\reachset}[1]{\mathcal{T}_{#1}}
\newcommand{\unknown}{\textbf{\textit{?}}}

\newcommand{\comp}[1]{\ensuremath{\overline{#1}}\xspace}

\newcommand{\prefixrel}{\sqsubseteq}

\newcommand{\myquot}[1]{``#1''}

\newcommand{\predmonitoring}[1]{{P}_{#1}}
\newcommand{\pospredmonitoring}[1]{{P}_{#1}^\top}
\newcommand{\negpredmonitoring}[1]{{P}_{#1}^\bot}

\makeatletter
\newcommand{\shortsim}{%
  \settowidth{\@tempdima}{n}% Width of n
  \resizebox{\@tempdima}{\height}{$\sim$}%
}
\makeatother

\newcommand{\Bsymbol}{\#}
\newcommand{\Dstate}{r}
\newcommand{\Estate}{s}
\newcommand{\Fstate}{t}